\documentclass[11pt]{article}
\usepackage{amssymb}
\usepackage{amsmath}
\usepackage{amsthm}
\usepackage{latexsym}
\usepackage{graphicx}
\newtheorem{theo}{Theorem}[section]
\newtheorem{prop}[theo]{Proposition}
\newtheorem{cor}[theo]{Corollary}

\theoremstyle{definition}
\newtheorem{defi}[theo]{Definition}
\newtheorem{exa}[theo]{Example}
\newtheorem{rem}[theo]{Remark}
\numberwithin{equation}{section}

\newcommand{\DS}{\displaystyle}

\newcommand{\N}{{\mathbb N}}
\newcommand{\F}{{\mathbb F}}

\newcommand{\Z}{{\mathbb Z}}
\newcommand{\Q}{{\mathbb Q}}
\newcommand{\C}{{\mathbb C}}

\newcommand{\cC}{{\mathcal C}}

\newcommand{\cH}{{\mathcal H}}

\newcommand{\cM}{{\mathcal M}}
\newcommand{\cP}{{\mathcal P}}
\newcommand{\cO}{{\mathcal O}}
\newcommand{\wcP}{\widehat{\mathcal P}}
\newcommand{\wwcP}{\widehat{\phantom{\big|}\hspace*{.5em}}\hspace*{-.9em}\wcP}

\newcommand{\cPhom}{{\mathcal P}_{\rm{hom}}}
\newcommand{\wcPhom}{\widehat{{\mathcal P}_{\rm{hom}}}}
\newcommand{\wcPhomr}{\widehat{{\mathcal P}_{\rm{hom}}}^{\scriptscriptstyle r}}
\newcommand{\wcPhoml}{\widehat{{\mathcal P}_{\rm{hom}}}^{\scriptscriptstyle l}}
\newcommand{\wcPchil}{\wcP^{^{\scriptscriptstyle[\chi,l]}}}
\newcommand{\wcPchir}{\wcP^{^{\scriptscriptstyle[\chi,r]}}}
\newcommand{\wcHr}{\widehat{\cH'}}
\newcommand{\cQ}{{\mathcal Q}}

\newcommand{\cL}{{\mathcal L}}

\newcommand{\veps}{{\varepsilon}}
\newcommand{\widesim}[1][1.5]{\scalebox{#1}[1]{$\sim$}}

\renewcommand{\mod}{\mbox{\rm mod}\,}
\newcommand{\soc}{\mbox{\rm soc}}
\newcommand{\rad}{\mbox{\rm rad}}

\newcommand{\wt}{{\rm wt}}

\newcommand{\mmid}{\mbox{$\,|\,$}}
\newcommand{\mmidbig}{\mbox{$\,\big|\,$}}

\newcommand{\inner}[1]{\mbox{$\langle{#1}\rangle$}}

\newcounter{alp}
\newcounter{ara}
\newcounter{rom}
\newenvironment{romanlist}{\begin{list}{(\roman{rom})\hfill}{\usecounter{rom}
     \topsep0ex \labelwidth.7cm \leftmargin.7cm \labelsep0cm
     \rightmargin0cm \parsep0ex \itemsep.4ex
     \partopsep1ex}}{\end{list}}
\newenvironment{alphalist}{\begin{list}{(\alph{alp})\hfill}{\usecounter{alp}
     \topsep.5ex \labelwidth.6cm \leftmargin.6cm \labelsep0cm
     \rightmargin0cm \parsep0ex \itemsep0ex
     \partopsep1.6ex}}{\end{list}}

\title{Partitions of Frobenius Rings Induced by the Homogeneous Weight}
\date\today
\author{Heide Gluesing-Luerssen\thanks{The author was partially supported by the National Science Foundation
        grants \#DMS-0908379 and \#DMS-1210061.}\\
        University of Kentucky\\ Department of Mathematics\\
       715 Patterson Office Tower\\ Lexington, KY 40506-0027, USA\\ heide.gl@uky.edu}

\begin{document}
\maketitle
\noindent{\bf Abstract:}
The values of the homogeneous weight are determined for finite Frobenius rings that are a direct product of local Frobenius rings.
This is used to investigate the partition induced by this weight and its dual partition under character-theoretic dualization.
A characterization is given of those rings for which the induced partition is reflexive or even self-dual.



\section{Introduction}\label{SS-Intro}
The homogeneous weight has been studied extensively in the literature of codes over rings.
It has been introduced by Constantinescu and Heise~\cite{CoHe97} as a generalization of both the Hamming weight
on finite fields and the Lee weight on~$\Z_4$.
Its main feature is that the average weight of the elements in a nonzero principal ideal is the same constant for all such ideals.
The weight has been further generalized to arbitrary non-commutative finite rings by Greferath and Schmidt~\cite{GrSch00} as well as
Honold and Nechaev~\cite{HoNe99}.
The homogeneous weight has proven to be an important tool in ring-linear coding.
For instance, in~\cite{DGLS01} Duursma et al.\ construct non-linear codes with the best parameters so far using certain
ring-linear codes and where the ring is endowed with the homogenous weight.

These and other properties of the homogeneous weight have led to a detailed study of this weight.
Among other things, it has been shown that the MacWilliams extension theorem remains true for isomorphisms preserving the homogeneous weight,
see Constantinescu et al.~\cite{CHH96} for codes over the integer residue ring~$\Z_N$, Wood~\cite{Wo97} and Greferath and
Schmidt~\cite{GrSch00} for codes over general finite Frobenius rings, and Greferath et al.~\cite{GNW04} for codes over
the Frobenius module of a finite ring.

On the other hand, so far no explicit MacWilliams identity for the homogeneous-weight enumerators of codes has
been established in any general form.
Such identities relate a suitably defined weight enumerator of a code to the weight enumerator (or a dual version thereof) of its dual code.
MacWilliams identities are well known for many weight functions, e.g., the Hamming weight, the complete weight, the symmetrized Lee weight, and more.

The non-existence of an explicit MacWilliams identity for the homogeneous weight is due to the fact that the partition induced by this
weight does not behave as well under dualization as those for the Hamming weight or the other weights just mentioned.
More precisely, the induced partition is in general not self-dual with respect to a certain character-theoretic dualization.
In this paper we will study the homogeneous weight partition on a certain class of finite Frobenius rings, which includes all commutative Frobenius rings,
and will provide a
characterization of those rings, for which the partition is reflexive (that is, coincides with its bidual) or even self-dual.
Reflexivity, which is weaker than self-duality, guarantees a MacWilliams identity because in this case the partition and its dual have
the same number of partition sets.
In this case we will also provide the associated Krawtchouk coefficients.
With these data, an explicit MacWilliams identity relating the corresponding partition enumerators is simply an instance of
the general theory about reflexive partitions, see for instance~\cite{GL13pos} or Camion~\cite{Cam98}.
A precursor of these ideas is the paper~\cite{BGO07} by Greferath et al.

We will study the homogeneous weight on Frobenius rings that are a direct product of local Frobenius rings.
For the latter the homogeneous weight is well-known~\cite{BGO07} and takes a very simple form.
Making use of an explicit formula for the values of the homogeneous weight provided by Honold~\cite{Hon01},
we will be able to compute the values of the homogeneous weight on the specified Frobenius rings.
This is carried out in Section~\ref{SS-homogWt}.

In Section~\ref{SS-HomogPart} we then go on and study the partition induced by the homogeneous weight.
We will see that all elements outside the socle have the same weight, thus form one partition set, whereas in the socle the homogeneous partition
is closely related to the product of Hamming partitions that are induced by a suitable direct product representation of the ring.
In fact, we will show that the homogeneous partition is reflexive if and only if its restriction to the socle coincides with the just described
product of the Hamming partitions.
We also give a characterization of the rings for which the homogeneous partition is reflexive.
It is given in terms of the orders of the residue fields of the local component rings.
Finally, we will prove that the homogeneous weight partition is self-dual if and only if it is reflexive and
the ring is semisimple.

\section{Frobenius Rings and Partitions}\label{SS-FrobPart}
Throughout this section, let~$G=(G,+)$ be a finite abelian group, and let~$R$ be a finite ring with unity.
Denote its group of units by~$R^*$.
Our main subject is partitions and their duals of~$R$ or~$R^n$, but occasionally we need to consider partitions of the group;
mainly for the situation where~$G$ is the additive group of an ideal of~$R$.
For this reason we present the main notions of this section for groups with special emphasis on rings.

Denote by $\widehat{G}$ the complex character group of~$G$.
Thus $\widehat{G}=\text{Hom}\big(G,\C^*\big)$ is the set of all group homomorphisms from~$G$ into~$\C^*$
with addition $(\chi_1+\chi_2)(a):=\chi_1(a)\chi_2(a)$.
The zero element is the \emph{principal character} $\veps\in\widehat{G}$, given by
$\veps(a)=1$ for all  $a\in G$.
It is well known that~$G$ and $\widehat{G}$ are isomorphic.
The most fundamental property of characters on~$G$ is the orthogonality relation
\begin{equation}\label{e-RChar}
   \sum_{a\in G}\chi(a)=\left\{\begin{array}{cl}
          0,&\text{if }\chi\neq\veps,\\ |G|,&\text{if }\chi=\veps.\end{array}\right.
\end{equation}

For a ring~$R$ and $n\in\N$ we denote by~$\widehat{R^n}$ the character group of the additive group $(R^n,+)$.
This group can be endowed with an $R$-$R$-bimodule structure via the left and right scalar multiplications
\begin{equation}\label{e-bimodule}
    (r\!\cdot\!\chi)(v)=\chi(vr)\ \text{ and }\
    (\chi\!\cdot\! r)(v)=\chi(rv)\text{ for all } r\in R\text{ and }v\in R^n.
\end{equation}
Recall that~$R$ is a Frobenius ring if $_R\soc(_R R)\cong\, _R(R/\rad(R))$,
where $\soc(_R R)$ denotes the socle of the left~$R$-module~$R$
and $\rad(R)$ is the Jacobson radical of~$R$
(it is well-known that the existence of a left isomorphism implies the right analogue).
Since $\rad(R)$ is a two-sided ideal, $R/\rad(R)$ is even a ring.
Furthermore, $\soc(_R R)=\soc(R_R)$, and we will simply write $\soc(R)$ for the socle.

We summarize the following properties about Frobenius rings, some of which actually characterize the
Frobenius property, but we will not engage in that discussion.
Details can be found in many books on ring theory, e.g.,~\cite[Ch.~6]{Lam99} by Lam
or in  the research articles by Lamprecht~\cite{Lamp53}, Hirano~\cite{Hi97}, Wood~\cite[Thm.~3.10]{Wo99},  and
Honold~\cite{Hon01}.

\begin{rem}\label{R-FrobProp}
Let~$R$ be a finite Frobenius ring. Then the following are true.
\begin{alphalist}
\item $R$ and~$\hat{R}$ are isomorphic left $R$-modules and isomorphic right $R$-modules.
	More precisely, there exists a character~$\chi\in\hat{R}$  such that
	\begin{equation}\label{e-RRhat}
	   R\longrightarrow \hat{R},\quad r\longmapsto r\!\cdot\!\chi,
	   \ \text{ resp.}\
	   R\longrightarrow \hat{R},\quad r\longmapsto \chi\!\cdot\! r
	\end{equation}
	is an isomorphism of left (resp.\ right) $R$-modules.
	Any such~$\chi$ is called a \emph{generating character} of~$R$.
	Obviously, any two generating characters~$\chi,\,\chi'$ differ by a unit, i.e., $\chi'=u\!\cdot\!\chi$ and
	 $\chi'=\chi\!\cdot\!u'$ for some $u,\,u'\in R^*$.
	More generally, for each~$n$, the maps
	\begin{equation}\label{e-RRnhat}
	     \alpha_l:\; R^n\longrightarrow \widehat{R^n},\quad v\longmapsto \chi(\inner{-,v}),
	     \ \text{ and }\
	     \alpha_r:\; R^n\longrightarrow \widehat{R^n},\quad v\longmapsto \chi(\inner{v,-}),
	\end{equation}
	are left (resp.\ right) $R$-module isomorphisms (here $\inner{v,w}$ denotes the standard inner product on~$R^n$).
\item Let~$\chi$ be a character of~$R$.
      Then~$\chi$ is a generating character of~$R$ if and only if the only left (resp.\ right) ideal contained in
      $\ker\chi:=\{a\in R\mid \chi(a)=1\}$ is the zero ideal; see~\cite[Cor.~3.6]{ClGo92}.
\item $R$ satisfies the double annihilator property, i.e.,
      $\text{ann}_l(\text{ann}_r(I))=I$ for each left ideal~$I$ of~$R$ and
      $\text{ann}_r(\text{ann}_l(I))=I$ for each right ideal~$I$,
      where $\text{ann}_l$ (resp.\ $\text{ann}_r$) denotes the left (resp.\ right) annihilator ideal;
      see \cite[Thm.~15.1]{Lam99}.
      In particular,  $\text{ann}_l(\rad(R))=\soc(R)=\text{ann}_r(\rad(R))$, see \cite[Cor.~15.7]{Lam99}.
\item $\soc(R)$ is a principal ideal~\cite[p.~21]{GrSch00}.
\item If~$R$ is a commutative finite Frobenius ring, then
       $R=R_1\times\ldots\times R_t$ for suitable local Frobenius rings~$R_i$; see~\cite[Th.~15.27]{Lam99}.
\item If~$R$ is local, that is, $\rad(R)$ is the unique maximal left (resp.\ right) ideal of~$R$, then~$\soc(R)$ is the unique minimal
      ideal~\cite[Ex.~(3.14)]{Lam99}.
      In this case $R/\rad(R)$ is called the \emph{residue field of} $R$.
\end{alphalist}
\end{rem}

Many standard examples of rings are Frobenius.
For details we refer to Wood~\cite[Ex.~4.4]{Wo99} and Lam~\cite[Sec.~16.B]{Lam99} and many other sources.
\begin{exa}\label{E-Frob}
\begin{alphalist}
\item The integer residue rings $\Z_N$, where $N\in\N$, are Frobenius.
      Each generating character is of the form $\chi(g):=\zeta^{g}$ for all $g\in\Z_N$, where
      $\zeta\in\C$ is an $N$-th primitive root of unity.
      Thus, each character is given by $(a\!\cdot\!\chi)(g)=\chi(ag)=\zeta^{ag}$ for some $a\in\Z_N$.
\item Every finite field is Frobenius, and every non-principal character is generating.
\item Finite chain rings (e.g., Galois rings),  finite group rings over a Frobenius ring and direct products of Frobenius rings
      are Frobenius.
\item Matrix rings over Frobenius rings are Frobenius.
\item The ring $R=\F_2[x,y]/(x^2,y^2,xy)$ is a local, non-Frobenius ring; see~\cite[Ex.~3.2]{ClGo92}.
\end{alphalist}
\end{exa}

\medskip
For a subgroup $H\leq G$ (called an \emph{additive code}) we define the \emph{dual subgroup} as
\begin{equation}\label{e-dualsubgroup}
   H^{\perp}:=\{\chi\in\widehat{G}\mid \chi(h)=1\text{ for all }h\in H\}.
\end{equation}
It is straightforward to see that
\begin{equation}\label{e-Hhat}
   \widehat{H}\cong \widehat{G}/H^{\perp}.
\end{equation}
As usual, a \emph{code over~$R$} is defined to be a left submodule of~$R^n$ for some~$n\in\N$.
Following the above, the \emph{dual code} is
$\cC^{\perp}=\{\chi\in\widehat{R^n}\mid \chi(v)=1\text{ for all }v\in\cC\}$.
It is easily seen that for any code~$\cC$ the dual $\cC^{\perp}$ is a right $R$-submodule of~$\widehat{R}$.
The double annihilator property from Remark~\ref{R-FrobProp}(c) extends to $\cC^{\perp\perp}=\cC$ and
$|\cC^{\perp}||\cC|=|R^n|$, see \cite[Cor.~5]{HoLa01} or  \cite[Rem.~5.5]{GL13pos}.
Applying the isomorphisms in~\eqref{e-RRnhat} to the additive group of~$\cC^{\perp}$ results in the right and left hand side dual
(see also \cite[Thm.~7.7]{Wo99})
\begin{equation}\label{e-Cdual}
\begin{split}
  \cC^{\perp,l}&:=\alpha_r^{-1}(\cC^{\perp})=\{v\in R^n\mid \inner{v,w}=0\text{ for all }w\in\cC\},\\
  \cC^{\perp,r}&:=\alpha_l^{-1}(\cC^{\perp})=\{v\in R^n\mid \inner{w,v}=0\text{ for all }w\in\cC\}.
\end{split}
\end{equation}
Note that $\cC^{\perp,l}$ is an $R$-$R$-bimodule, whereas $\cC^{\perp,r}$ is in general just a right $R$-module.

\medskip
We now turn to partitions on~$G$ and fix the following notation.
A partition $\cP=(P_m)_{m=1}^M$ of a set~$X$ will mostly be
written as $\cP=P_1\mmid P_2\mmid\ldots\mmid P_M$.
The sets of the partition are called its \emph{blocks}, and we write $|\cP|$ for the number of blocks in~$\cP$.
Recall that two partitions~$\cP$ and~$\cQ$ are called \emph{identical} if $|\cP|=|\cQ|$ and the blocks coincide
after suitable indexing.
Moreover,~$\cP$ is called \emph{finer} than~$\cQ$ (or~$\cQ$ is \emph{coarser} than~$\cP$), written as $\cP\leq\cQ$, if
every block of~$\cP$ is contained in a block of~$\cQ$.
Note that if~$\cP\leq\cQ$ then $|\cP|\geq|\cQ|$.
Denote by~$\widesim_{\cP}$ the equivalence relation induced by~$\cP$, thus,
$v\widesim_{\cP}v'$ if $v,\,v'$ are in the same block of~$\cP$.

The following notion of a dual partition will be crucial for us.
The left-sided version has been introduced for Frobenius rings by Byrne et~al.~\cite[p.~291]{BGO07} and goes back to the notion of
F-partitions as introduced by Zinoviev and Ericson in~\cite{ZiEr96}.
Reflexive partitions, defined below, appear already in~\cite{ZiEr09} by Zinoviev and Ericson  were they have been
coined B-partitions.
They are exactly the partitions that induce abelian association schemes as
studied in a more general context by Delsarte~\cite{Del73}, Camion~\cite{Cam98}, and others, see also~\cite{DeLe98}.
For an overview of these various approaches and their relations in the language of partitions, see also~\cite{GL13pos}.

Throughout, we will use the notation $[n]:=\{1,\ldots,n\}$ and $[n]_0:=\{0,\ldots,n\}$.

\begin{defi}\label{D-DualPart}
\begin{alphalist}
\item Let $\cP=P_1\mmid P_2\mmid\ldots\mmid P_M$ be a partition of the group~$G$.
The \emph{dual partition}, denoted by~$\wcP$, is the partition of~$\widehat{G}$ defined via the equivalence relation
\begin{equation}\label{e-simPhat}
  \chi\widesim_{\wcP} \chi' :\Longleftrightarrow \sum_{g\in P_m}\chi(g)=\sum_{g\in P_m}\chi'(g)
  \text{ for all }m\in[M].
\end{equation}
For the pair $(\cP,\wcP)$, where $\wcP=Q_1\mmid\ldots\mmid Q_L$,
the \emph{generalized Krawtchouk coefficients}~$K_{\ell,m},\,\ell\in[L],\,m\in[M]$, are defined as
$K_{\ell,m}=\sum_{g\in P_m}\chi(g)$, where $\chi$ is any element in $Q_\ell$.
The partition~$\cP$ is called \emph{reflexive} if $\wwcP=\cP$.
\item Let~$R$ be a Frobenius ring with generating character~$\chi$, and let
$\cP=P_1\mmid P_2\mmid\ldots\mmid P_M$ be a partition of the group~$(R^n,+)$.
The \emph{left and right $\chi$-dual partition} of~$\cP$, denoted
by~$\wcP^{^{\scriptscriptstyle[\chi,l]}}$ and ~$\wcP^{^{\scriptscriptstyle[\chi,r]}}$, are defined as the preimage of~$\wcP$
under the isomorphisms~$\alpha_l$ and~$\alpha_r$ in~\eqref{e-RRnhat}.
Thus, $\wcP^{^{\scriptscriptstyle[\chi,l]}}$ and $\wcP^{^{\scriptscriptstyle[\chi,r]}}$ are the partitions of~$R^n$ given by
the equivalence relations
\begin{equation}\label{e-simPhat2l}
  v\widesim_{\wcP^{^{\scriptscriptstyle[\chi,l]}}}\, v' :\Longleftrightarrow
  \sum_{w\in P_m}\chi(\inner{w,v})=\sum_{w\in P_m}\chi(\inner{w,v'}) \text{ for all }m=1,\ldots,M
\end{equation}
and
\begin{equation}\label{e-simPhat2r}
  v\widesim_{\wcP^{^{\scriptscriptstyle[\chi,r]}}}\, v' :\Longleftrightarrow
  \sum_{w\in P_m}\chi(\inner{v,w})=\sum_{w\in P_m}\chi(\inner{v',w}) \text{ for all }m=1,\ldots,M,
\end{equation}
respectively.
\end{alphalist}
\end{defi}

It is not hard to find examples of rings and partitions for which the left and right 
$\chi$-dual partitions do not coincide for any generating character~$\chi$ (take, e.g., the ring in \cite[Ex.~1.4(iii)]{Wo99}).
We  have the following relation.
\begin{prop}[\mbox{\cite[Prop.~4.4]{BGL13}}]\label{P-leftrightbidual}
Let~$\cP$ be a partition on~$R^n$.
Then
\begin{equation}\label{e-bidual}
 \widehat{\phantom{\Big|}\hspace*{1.4em}}^{\,\scriptscriptstyle[\chi,r]}\hspace*{-3.5em}\widehat{\cP}^{^{\scriptscriptstyle[\chi,l]}}
 \hspace*{1.2em}
 =\wwcP
 =\;\widehat{\phantom{\Big|}\hspace*{1.4em}}^{\,\scriptscriptstyle[\chi,l]}\hspace*{-3.3em}\widehat{\cP}^{^{\scriptscriptstyle[\chi,r]}}\hspace*{1.2em},
\end{equation}
where~$\wwcP$ is the bidual partition in the group sense of~\eqref{e-simPhat}.
Consequently,~$\cP$ is reflexive if and only if
$\cP=\;\widehat{\phantom{\Big|}\hspace*{1.4em}}^{\,\scriptscriptstyle[\chi,r]}\hspace*{-3.5em}\widehat{\cP}^{^{\scriptscriptstyle[\chi,l]}}\hspace*{1.2em}$, which is equivalent to
$\cP=\;\widehat{\phantom{\Big|}\hspace*{1.4em}}^{\,\scriptscriptstyle[\chi,l]}\hspace*{-3.5em}\widehat{\cP}^{^{\scriptscriptstyle[\chi,r]}}\hspace*{1.2em}$.
Moreover, $\cP=\wcPchil$ if and only if $\cP=\wcPchir$.
We call~$\cP$ $\chi$-self-dual if $\cP=\wcPchir$.
\end{prop}

Let us briefly discuss the Krawtchouk coefficients $K_{\ell,m}$ of the pair $(\cP,\wcP)$ from Definition~\ref{D-DualPart}(a).
Due to the very definition of the dual partition, these coefficients do not depend on the choice of~$\chi$ in~$Q_{\ell}$.
In the ring setting the coefficients read as follows.
Let $\wcP^{^{\scriptscriptstyle[\chi,l]}}=Q'_1\mmid\ldots\mmid Q'_L$ and
$\wcP^{^{\scriptscriptstyle[\chi,r]}}=Q''_1\mmid\ldots\mmid Q''_L$.
Hence $Q'_m=\alpha_l^{-1}(Q_m)$ and $Q''_m=\alpha_r^{-1}(Q_m)$ for $m=1,\ldots,L$.
Then~\eqref{e-RRnhat} yields
\begin{eqnarray}
    K_{\ell,m}&=\sum_{w\in P_m}\chi(\inner{w,v}),\text{ where~$v$ is any element in }Q'_\ell, \label{e-klm2l}\\
                   &=\sum_{w\in P_m}\chi(\inner{v,w}),\text{ where~$v$ is any element in }Q''_\ell \label{e-klm2r}.
\end{eqnarray}
In particular,~$K_{\ell,m}$ does not depend on the choice of~$v$ inside~$Q'_\ell$ (resp.~$Q''_\ell$), and neither does
it  depend on the sidedness of the dual partition.

\medskip
The following example illustrates that the dual partition does in general depend on the choice of the generating character~$\chi$, even in the commutative case.
\begin{exa}\label{E-FPartNotIndep}
Consider the field~$\F_4=\{0,\,1,\,a,\,a^2\}$, where $a^2=a+1$.
The maps~$\chi,\,\tilde{\chi}$ given by
\[
  \chi(0)=\chi(1)=1,\,\chi(a)=\chi(a^2)=-1 \;\text{ and }\;
  \tilde{\chi}(0)=\tilde{\chi}(a)=1,\,\tilde{\chi}(1)=\tilde{\chi}(a^2)=-1,
\]
are characters of~$\F_4$.
As in~\eqref{e-RRhat} they give rise to the isomorphisms $\alpha:\F_4\rightarrow\widehat{\F_4},\, r\mapsto r\!\cdot\!\chi$ and
$\beta:\F_4\rightarrow\widehat{\F_4},\, r\mapsto r\!\cdot\!\tilde{\chi}$.
Consider the partition~$\cP$ of~$\F_4$ given by $P_0\mmid P_1\mmid P_2=0\mmid1\mmid a,a^2$.
Then $\sum_{b\in P_j}\chi(ab)=\sum_{b\in P_j}\chi(a^2b)=(-1)^j$ for $j=0,1$, while
$\sum_{b\in P_2}\chi(ab)=\chi(a^2)+\chi(1)=0$ and $\sum_{b\in P_2}\chi(a^2b)=\chi(1)+\chi(a)=0$.
Writing simply~$\wcP^{^{\scriptscriptstyle[\chi]}}$ for $\wcP^{^{\scriptscriptstyle[\chi,r]}}=\wcP^{^{\scriptscriptstyle[\chi,l]}}$, we conclude
$\wcP^{^{\scriptscriptstyle[\chi]}}=\cP$ and~$\cP$ is $\chi$-self-dual, hence also reflexive.
In the same way one computes $\wcP^{^{\scriptscriptstyle[\tilde{\chi}]}}$ and obtains
$\cQ:=\wcP^{^{\scriptscriptstyle[\tilde{\chi}]}}=0\mmid1,a^2\mmid a$.
Hence $\cP$ is not $\tilde{\chi}$-self-dual.
One easily verifies $\widehat{\cQ}^{^{\scriptscriptstyle[\tilde{\chi}]}}=\cP$, which also follows from the fact
that reflexivity does not depend on the choice of the generating character.
\end{exa}

\begin{rem}\label{R-ZNcharac}
For the integer residue rings $R=\Z_N$, the dual of a partition does not depend on the choice of
the generating character.
This follows immediately from Example~\ref{E-Frob}(a) along with the fact that all primitive $N$-th roots of unity
have the same minimal polynomial.
The latter implies  that the identities on the right hand side of~\eqref{e-simPhat2l} do not depend on the choice of the primitive
$N$-th root~$\zeta$.
\end{rem}

For many standard partitions on Frobenius rings, e.g., the Hamming partition, the dual does not depend on the choice of the
generating character.
In the next section we will see that this is also the case for the main topic of this paper, partitions induced by the homogeneous weight.

Let us return to the general situation of partitions of groups.
We have the following simple observation.

\begin{rem}\label{R-Dual0}
\begin{alphalist}
\item The singleton~$\{\veps\}$ is always a block of~$\wcP$.
      Indeed, $\sum_{g\in P_m}\veps(g)=|P_m|$ for every block~$P_m$.
      Hence if~$\chi\in\widehat{G}$ satisfies $\sum_{g\in P_m}\chi(g)=|P_m|$ for all $P_m$, then
      $\sum_{g\in G}\chi(g)=|G|$ and thus $\chi=\veps$ by~\eqref{e-RChar}.
      Using~\eqref{e-RRnhat} we also conclude that $\{0\}$ is a block of $\wcP^{^{\scriptscriptstyle[\chi,l]}}$ and
      $\wcP^{^{\scriptscriptstyle[\chi,r]}}$.
\item If $\cP\leq\cQ$, then $\wcPchil\leq\widehat{\cQ}^{^{\scriptscriptstyle[\chi,l]}}$
      and $\wcPchir\leq\widehat{\cQ}^{^{\scriptscriptstyle[\chi,r]}}$.
      This follows directly from the definition of the dual partitions since each block of~$\cQ$ is the union of blocks of~$\cP$.
\end{alphalist}
\end{rem}

As has been shown in various forms in the literature~\cite{Del73,ZiEr09,BGO07,GL13pos}, a partition~$\cP$
of~$R^n$ and its dual partition $\wcP$ of $\widehat{R^n}$ allow a MacWilliams identity:
applying a certain MacWilliams transformation to the $\wcP$-partition enumerator
of a code $\cC\subseteq\widehat{R^n}$ results in the $\cP$-enumerator of its dual~$\cC^\perp\subseteq R^n$.
For an overview in the language of this paper see~\cite[Sec.~2]{GL13pos}.
The most symmetric situation arises for reflexive (or even self-dual) partitions, in which case the transformation can be carried
out in both directions, and thus the two enumerators determine each other uniquely.
Most, if not all, classical examples of MacWilliams identities are instances of this general MacWilliams identity based on a self-dual
partition (for instance, for the Hamming weight, symmetrized Lee weight, complete weight, and the rank metric).

In the next section we will study the partition on a Frobenius ring induced by the homogeneous weight and investigate for which rings
the partition is reflexive or even self-dual.
Our main tool for characterizing reflexivity of a partition is the following convenient criterion from~\cite[Thm.~3.1]{GL13pos}.

\begin{theo}\label{T-ReflCrit}
For any partition~$\cP$ on~$G$ and its dual partition~$\wcP$
we have $|\cP|\leq|\wcP|$ and
$\widehat{\phantom{\big|}\hspace*{.6em}}\hspace*{-.9em}\wcP\leq\cP$.
Moreover, $\cP$ is reflexive if and only if $|\cP|=|\wcP|$.
\end{theo}

We close this section with several specific instances of reflexive partitions that will be needed in the next section.

\begin{exa}\label{E-Punits}
Let~$R$ be a Frobenius ring with generating character~$\chi$, and denote by $\cP^{*,l}$ (resp.\ $\cP^{*,r}$) the partition
given by the orbits of the left (resp.\ right) action of~$R^*$ on~$R$.
Thus, the blocks of~$\cP^{*,l}$ are given by the distinct orbits $\cO_{x,l}=\{ux\mid u\in R^*\},\,x\in R$,
whereas the blocks of~$\cP^{*,r}$  are given by the orbits $\cO_{x,r}=\{xu\mid u\in R^*\}$.
It follows from \cite[Prop.~4.6]{BGL13} that
$\widehat{\cP^{*,l}}^{\scriptscriptstyle[\chi,r]}=\cP^{*,r}$ and
$\widehat{\cP^{*,r}}^{\scriptscriptstyle[\chi,l]}=\cP^{*,l}$.
In particular, the dual partitions do not depend on the choice of~$\chi$, and the partitions are reflexive.
\end{exa}

The Hamming partition and its Krawtchouk coefficients will be needed in the following form in the next section.

\begin{exa}\label{E-HammPart}
Let~$G=A_1\times\ldots\times A_n$, where~$A_i$ is a finite abelian group for all~$i$.
Then $\widehat{G}=\widehat{A_1}\times\ldots\times\widehat{A_n}$, i.e., the characters of~$G$ are given by
\[
    \big(\chi_1,\ldots,\chi_n\big)(a_1,\ldots,a_n)=\prod_{i=1}^n\chi_i(a_i).
\]
Assume $|A_i|=q$ for all~$i$, and let~$\cP$ be the partition of~$G$ induced by the Hamming weight, that is,
$P_m=\{a\in G\mid \wt(a)=m\}$ for $m\in[n]_0=\{0,\ldots,n\}$, and where $\wt(a_1,\ldots,a_n)=|\{i\mid a_i\neq0\}|$.
Then it is well known that the dual partition~$\wcP$ is the partition induced by the Hamming weight on
$\widehat{A_1}\times\ldots\times\widehat{A_n}$.
In other words, the blocks of~$\wcP$ are
$Q_\ell=\{\chi\in\widehat{G}\mid \wt(\chi)=\ell\}$ for $\ell\in[n]_0$, and where
$\wt(\chi_1,\ldots,\chi_n)=|\{i\mid \chi_i\neq\veps\}|$ since~$\veps$ is the zero element of~$\widehat{G}$.
The Krawtchouk coefficients are  $K_{\ell,m}=\sum_{a\in P_m}\chi(a)=K_m^{(n,q)}(\ell)$ for $\chi\in Q_\ell$, and
where
\begin{equation}\label{e-Krawtclassic}
          K_m^{(n,q)}(x)=\sum_{j=0}^m(-1)^j(q-1)^{m-j}{x\choose j}{n-x\choose m-j}
\end{equation}
is the Krawtchouk polynomial (see for instance \cite[Thm.~4.1]{Del73} or Lemma~2.6.2 in~\cite{HP03};
the proof of Lemma~2.6.2 in~\cite{HP03}, given for $\Z_q^n$, works mutatis mutandis for all
$A_1\times\ldots\times A_n$).
All of this shows that the Hamming partition on~$G$ is reflexive.
Finally, the Hamming partition on the module~$R^n$ is self-dual with respect to any generating character~$\chi$ by virtue
of~\eqref{e-RRnhat}.
\end{exa}

When studying the partition induced by the homogeneous weight, we will need to consider product partitions.
The following result will suffice for us.
Let $G=A_1\times\cdots\times A_n$, where $A_1,\ldots,A_n$ are finite abelian groups and let~$\cP_i$  be
reflexive partitions of~$A_i$ for $i=1,\ldots,n$.
Write $\cP_i=P_{i,1}\mmid P_{i,2}\mmid\ldots\mmid P_{i,M_i}$.
Then the \emph{product partition} on~$G$  is defined as
\begin{equation}\label{e-prodPart}
    \cQ:=\cP_1\times\ldots\times\cP_n:=(P_{1,m_1}\times\ldots\times P_{n,m_n})_{(m_1,\ldots,m_n)
  \in[M_1]\times\cdots\times[M_n]}.
\end{equation}
Parts of the next result also appeared in \cite[Thm.~4]{ZiEr09} and \cite[Thm.~4.86]{Cam98}.

\begin{theo}[\mbox{\cite[Thm.~4.3 and proof]{GL13pos}}]\label{T-IndProdPart}
The dual partition of~$\cQ$ is
$\widehat{\cQ}=\widehat{\cP_1}\times\cdots\times\widehat{\cP_n}$, and in particular,~$\cQ$ is reflexive.
Furthermore, the Krawtchouk coefficients of $(\cQ,\widehat{\cQ})$ are
\[
    K_{(\ell_1,\ldots,\ell_n),(m_1,\ldots,m_n)}=\prod_{i=1}^n K^{(i)}_{\ell_i,m_i} \text{ for all }
    (\ell_1,\ldots,\ell_n),\,(m_1,\ldots,m_n)\in[M_1]\times\cdots\times[M_n],
\]
where $K^{(i)}_{\ell,m},\,\ell,m\in[M_i]$, are the Krawtchouk coefficients of $(\cP_i,\widehat{\cP_i})$.
\end{theo}

The following result shows that the trivial extension of a partition of a subgroup behaves well under dualization.
The notation $\chi_{|H}$ stands for the restriction of~$\chi$ to~$H$, whereas~$\veps_G$ and~$\veps_H$ denote the principal
characters on~$G$ and~$H$, respectively.
\begin{prop}\label{P-PartSubgroup}
Let $H\leq G$ be a subgroup of~$G$, and let $\cP=P_0\mmid\ldots\mmid P_M$ be a partition of~$H$.
Let $\wcP=Q_0\mmid Q_1\mmid\ldots\mmid Q_L$, where $Q_0=\{\veps_H\}$ (see Remark~\ref{R-Dual0}(a)).
Thus~$\wcP$ is a partition of~$\widehat{H}$.
Define $P_{-1}:=G\,\backslash\, H$.
Then $\cP'=P_{-1}\mmid P_0\mmid\ldots\mmid P_M$ is a partition of~$G$.
The dual partition of~$\cP'$ is given by $\widehat{\cP'}=Q'_{-1}\mmid Q'_0\mmid Q'_1,\ldots\mmid Q'_L$, where
$Q'_{-1}=H^{\perp}\backslash\{\veps_G\},\, Q'_0=\{\veps_G\}$ and
$Q'_\ell=\{\chi\in\widehat{G}\,\backslash\, H^{\perp}\mid \chi_{|H}\in Q_\ell\}$ for $\ell\in[L]$.
The Krawtchouk coefficients of $(\cP',\widehat{\cP'})$ are
\begin{equation}\label{e-Kraw1}
    (K'_{\ell,m})_{\ell=-1,\ldots,L\atop m=-1,\ldots,M}=
    \begin{pmatrix}-|H|&|P_0|&\cdots &|P_M|\\ |P_{-1}|& K_{0,0}&\ldots & K_{0,M}\\
                       0& K_{1,0}&\ldots & K_{1,M}\\ \vdots&\vdots & &\vdots\\ 0& K_{L,0}&\ldots & K_{L,M}\end{pmatrix},
\end{equation}
where $K_{\ell,m},\,\ell\in[L]_0,\,m\in[M]_0$, are the Krawtchouk coefficients of $(\cP,\wcP)$.
As a consequence, if~$\cP$ is reflexive then so is $\cP'$.
\end{prop}
\begin{proof}
We have to consider various cases.
\\
1) Let $\chi\in Q'_\ell$ for $\ell\in[L]$.
Then $\sum_{a\in P_m}\chi(a)=\sum_{a\in P_m}\chi_{|H}(a)=K_{\ell,m}$ for each for $m\in[M]_0$.
Furthermore,~\eqref{e-RChar} yields $\sum_{a\in G\backslash H}\chi(a)=-\sum_{a\in H}\chi(a)=-\sum_{a\in H}\chi_{|H}(a)$.
But the latter is zero due to~\eqref{e-RChar} since $\chi_{|H}$ is not the principal character on~$H$.
\\
2) Next, let $\chi\in Q'_{-1}=H^{\perp}\backslash\{\veps_G\}$. Then $\chi(a)=1$ for all $a\in P_m,\,m\in[M]_0$.
Thus $\sum_{a\in P_m}\chi(a)=|P_m|$ for all $m\in[M]_0$.
Moreover,  $\sum_{a\in G\backslash H}\chi(a)=-\sum_{a\in H}\chi(a)=-|H|$.
\\
3) For $\chi=\veps_G$ we obviously have $\sum_{a\in P_m}\chi(a)=|P_m|$ for all $m\in\{-1,0,\ldots,M\}$.
For the same reason,  $|P_m|=K_{0,m}$ for all $m\in[M]_0$.
\\
Evidently, in all of these cases the sums do not depend on the specific choice of~$\chi$ within the specified set,
and thus the partition $\cQ:=Q'_{-1}\mmid Q'_0\mmid Q'_1,\ldots\mmid Q'_L$ is finer than or equal to
$\widehat{\cP'}$.
The above also establishes the Krawtchouk coefficients stated in~\eqref{e-Kraw1}.
Since~$\wcP$ is the dual partition of~$\cP$, Definition~\ref{D-DualPart}(a) implies that no two rows of the matrix
in~\eqref{e-Kraw1} coincide.
This means that  if $\chi\in Q'_\ell$ and $\chi'\in Q'_{\ell'}$, where $\ell\neq\ell'$, then
$\chi\;\not\hspace*{-.5em}\widesim_{\widehat{\cP'}}\chi'$.
Thus $\widehat{\cP'}=\cQ$, as desired.
The statement concerning reflexivity follows from Theorem~\ref{T-ReflCrit}.
\end{proof}

\section{Explicit Values of the Homogeneous Weight}\label{SS-homogWt}
In this section we consider the homogeneous weight and determine its values
for those finite Frobenius rings that are isomorphic to a product of local rings.
Due to Remark~\ref{R-FrobProp}(e) this includes all finite commutative Frobenius rings.
In the subsequent section, the results will be used to study the partition induced by the homogeneous weight.

Throughout, let~$R$ be a finite Frobenius ring with group of units~$R^*$, and fix a generating character~$\chi$.
The following definition is taken from  Greferath and Schmidt~\cite{GrSch00}.
\begin{defi}\label{D-homogWt}
The \emph{(left) homogeneous weight on~$R$ with average value}~$\gamma$ is a function
$\omega:R\longrightarrow\Q$ such that
\begin{romanlist}
\item $\omega(0)=0$,
\item $\omega(x)=\omega(y)$ for all $x,y\in R$ such that $Rx=Ry$ ,
\item $\sum_{y\in Rx}\omega(y)=\gamma|Rx|$ for all $x\in R\backslash\{0\}$; in other words, the average weight over each
      nonzero principal ideal is~$\gamma$.
\end{romanlist}
\end{defi}
In~\cite[Thm.~1.3]{GrSch00} Greferath/Schmidt proved the existence and uniqueness of the homogeneous weight with given average value for any finite ring, and in~\cite[Cor.~1.6]{GrSch00} the same authors show that~(iii) is satisfied by the homogeneous weight for all nonzero ideals of~$R$ (for this result the Frobenius property is essential).

It is easy to see that the Hamming weight on~$R$ is homogeneous if and only if~$R$ is a field, in which case it has average value $\frac{q-1}{q}$, where $q=|R|$.

Without loss of generality we restrict ourselves to the homogeneous weight with average value $\gamma=1$, which we call
the \emph{normalized homogeneous weight}.
Thanks to Honold~\cite[p.~412]{Hon01}, an explicit formula for the normalized homogeneous weight on a Frobenius
ring is known and reads as
\begin{equation}\label{e-HomWt}
   \omega(r)=1-\frac{1}{|R^*|}\sum_{u\in R^*}\chi(ru)  \text{ for } r\in R;
\end{equation}
see \cite[Prop.1.3]{GrO04} for a short proof verifying that the function satisfies Definition~\ref{D-homogWt}(i)~--~(iii).
A specific instance of this formula, tailored to Galois rings, appears also in~\cite{VoWa03}.
It follows from~\eqref{e-HomWt} that the left homogeneous weight is also right homogeneous, i.e., it satisfies
the right counterparts of Definition~\ref{D-homogWt}(ii) and~(iii); see \cite[Th.~2]{Hon01}.
As already observed in the proof of \cite[Th.~2]{Hon01}, we also have
\begin{equation}\label{e-HomWtleft}
     \omega(r)=1-\frac{1}{|R^*|}\sum_{u\in R^*}\chi(ur)  \text{ for } r\in R.
\end{equation}
This follows simply from the bimodule structure of~$\widehat{R}$, see~\eqref{e-bimodule}, along with the fact that all generating characters of~$R$ are given by $u\!\cdot\!\chi$, where $u\in R^*$,  as well as by $\chi\!\cdot\!u$, where $u\in R^*$.

Identities~\eqref{e-HomWt} and~\eqref{e-HomWtleft} yield
\begin{equation}\label{e-KrawHom}
  \sum_{u\in R^*}\chi(ru)=\sum_{u\in R^*}\chi(ur)=|R^*|(1-\omega(r)) \text{ for }r\in R,
\end{equation}
which will be useful later for determining the values of the homogeneous weight on arbitrary Frobenius rings.

\begin{defi}\label{D-Phom}
The partition of~$R$ induced by homogeneous weight is denoted by~$\cPhom$.
It is thus given by the equivalence relation
$x\widesim_{\cPhom} x'\Longleftrightarrow \omega(x)=\omega(x') \text{ for }x,\,x'\in R$.
In other words, the blocks of~$\cPhom$ are the sets consisting of all elements sharing the same homogeneous weight.
\end{defi}

Recall the partitions $\cP^{*,l}$ and~$\cP^{*,r}$ induced by the left and right action of~$R^*$ on~$R$, see
Example~\ref{E-Punits}.
Definition~\ref{D-homogWt}(ii) implies $\cP^{*,l}\leq\cPhom$ and~\eqref{e-HomWt} yields
$\cP^{*,r}\leq\cPhom$.
Thus,  $\cP^{*,r}\leq\widehat{\cPhom}^{\scriptscriptstyle[\chi,r]}$ and
$\cP^{*,l}\leq\widehat{\cPhom}^{\scriptscriptstyle[\chi,l]}$ due to Example~\ref{E-Punits} and Remark~\ref{R-Dual0}(b).

Using the fact that the orbit~$\cO_{x,r}$ of~$x\in R$ is given by $\cO_{x,r}=\{xu_1,\ldots,xu_m\}$, where
$u_1,\ldots,u_m$ are representatives of the distinct right cosets of the stabilizer subgroup of~$x\in R$, one obtains
\begin{equation}\label{e-OrbitUnits}
   \sum_{b\in \cO_{x,r}}\chi(ab)=\frac{|\cO_{x,r}|}{|R^*|}\sum_{u\in R^*}\chi(axu)\text{ for all }a,\,x\in R.
\end{equation}

Now we can summarize the following simple properties, which will be needed later when computing the dual partition
of~$\cPhom$.
The latter will be defined as the partition $\widehat{\cPhom}^{\scriptscriptstyle[\chi,r]}$ of~$R$ in the sense of
Definition~\ref{D-DualPart}(b).
We will now see that the dual does not depend on the choice of~$\chi$.

\begin{rem}\label{R-UnitOrbits}
Consider a block~$P$ of the partition~$\cPhom$.
Then $P=\bigcup_{i=1}^M\hspace*{-2em}\raisebox{.5ex}{$\cdot$}\hspace*{1.7em}\cO_{x_i,r}$ for a certain
number~$M$ of distinct orbits~$\cO_{x_i,r}$.
With the aid of~\eqref{e-OrbitUnits} and~\eqref{e-KrawHom} we obtain for any $a\in R$
         \begin{align*}
            \sum_{b\in P}\chi(ab)&=\sum_{i=1}^M\sum_{b\in\cO_{x_i,r}}\chi(ab)
                 =\sum_{i=1}^M \frac{|\cO_{x_i,r}|}{|R^*|}\sum_{u\in R^*}\chi(ax_iu)\\
               &=\frac{1}{|R^*|}\sum_{i=1}^M|\cO_{x_i,r}|\big(|R^*|(1-\omega(ax_i)\big)
                 =\sum_{i=1}^M|\cO_{x_i,r}|(1-\omega(ax_i)).
         \end{align*}
This shows that the sum $\sum_{b\in P}\chi(ab)$ does not depend on the choice of the generating character~$\chi$.
As a consequence, the dual partition $\widehat{\cPhom}^{\scriptscriptstyle[\chi,r]}$ does not depend  on~$\chi$.
We will therefore simply denote this partition by~$\wcPhomr$.
Thus
\[
  a\widesim[2]_{\wcPhomr\,}a'\Longleftrightarrow \sum_{b\in P}\chi(ab)=\sum_{b\in P}\chi(a'b)
  \text{ for all blocks $P$ of }\cPhom.
\]
Note also that the Krawtchouk coefficients of $(\cPhom,\wcPhomr)$ are real.
One may also observe that the above implies $\cP^{*,l}\leq\wcPhomr$.
\end{rem}

Similar statements can be obtained for the left dual of~$\cPhom$, which then also does not depend on~$\chi$.

The following examples illustrate that the homogeneous partition on a Frobenius ring may display a variety of different properties.
Since all examples are commutative we will simply write $\wcPhom$ instead of $\wcPhomr$.

\begin{exa}\label{E-Phom}
\begin{alphalist}
\item On~$\Z_8$ the homogeneous partition is given by $\cPhom=0\mid 1,2,3,5,6,7\mid 4$.
        Its dual is $\wcPhom=0\mid 1,3,5,7\mid 2,4,6$, which has been observed already in~\cite[Ex.~2.9]{BGO07}.
        Thus $\cPhom$ is not self-dual, but reflexive, due to Theorem~\ref{T-ReflCrit}.
\item On $R=\Z_2\times\Z_2$ the homogeneous partition is easily seen to be $00,11\mid 01,10$, and the values of the weight
         are~$0$ and~$2$, respectively.
         The fact that $\{00\}$ is not a block of~$\cPhom$ shows that~$\cPhom$ is not the dual of any partition;
        see Remark~\ref{R-Dual0}(a).
        In particular,~$\cPhom$ is not reflexive and thus not self-dual.
        In Remark~\ref{R-HomWtValues}(b) a characterization will be presented for the Frobenius rings that contain a
        nonzero element with homogeneous weight zero.
\item Consider $R=\Z_3\times\Z_3$.
       In this case one easily verifies that $\cPhom=00\mmid 10,20,01,02\mmid 11,12,21,22$ with normalized homogeneous weights
       $0,\,3/2$, and~$3/4$, respectively.
       Note that~$\cPhom$ is simply the Hamming partition on~$\Z_3\times\Z_3$, and thus reflexive and even
       self-dual, i.e., $\cPhom=\wcPhom$.
       This will also follow from Theorem~\ref{T-PhomAdmiss}.
\item Consider $\F_4=\{0,1,\alpha,\alpha^2\}$ and the ring $R=\Z_2\times\F_4$.
        Using for instance the character~$\chi$ on~$\F_4$ given in
         Example~\ref{E-FPartNotIndep}
        one derives
        $\cPhom=00\mid 10\mid  01, 0\alpha,0\alpha^2\mid 11,1\alpha,1\alpha^2$, and the values of the weight are
        $0,\,2,\,4/3$, and $2/3$, respectively.
        Again, it will follow from Theorem~\ref{T-PhomAdmiss} that $\cPhom$ is self-dual and thus reflexive.
\end{alphalist}
\end{exa}

In order to determine the values of the homogeneous weight explicitly, we start with the following well-known case.
\begin{exa}\label{E-LocFrob}
Let~$R$ be a local Frobenius ring with residue field $R/\rad(R)$ of order~$q$.
Then the normalized homogeneous weight is given by
\[
     \omega(a)=\left\{\begin{array}{cl}
                                     0,&\text{if }a=0,\\ \frac{q}{q-1},&\text{if } a\in \soc(R)\backslash\{0\},\\1,&\text{otherwise}.
                                  \end{array}\right.
\]
This can be verified immediately using the fact that $\soc(R)$ is the unique minimal left ideal~\cite[Ex.~(3.14)]{Lam99}
and thus contained in any nonzero left ideal of~$R$;
see also~\cite[Ex.~2.8]{BGO07} for an argument involving the M\"obius function for the lattice of ideals of~$R$.
Hence $|\cPhom|=3$.
In Theorem~\ref{T-PhomAdmiss} (see also Example~\ref{E-homogWt2}(a)) we will see that $\cPhom$ is reflexive, and we
will also determine the dual partition.
\end{exa}

We now proceed to determine the values of the homogeneous weight for
finite Frobenius rings that can be written as the direct product of local Frobenius rings.
Thanks to Remark~\ref{R-FrobProp}(e),  this covers all finite commutative Frobenius rings.
We first summarize some basic properties for direct products of Frobenius rings.

\begin{rem}\label{R-ProdFrob}
Let $R=R_1\times\ldots\times R_t$, where each~$R_i$ is a finite (not necessarily local)  Frobenius ring.
Furthermore, let~$\chi_i$ be a generating character of~$R_i$ for all $i\in[t]$.
Then~$R$ is a Frobenius ring with character module $\widehat{R}=\widehat{R_1}\times\ldots\times\widehat{R_t}$; see also
Example~\ref{E-HammPart} for the group version.
A generating character~$\chi$ of~$R$ is given by $\chi:=(\chi_1,\ldots,\chi_t)$ defined as
\begin{equation}\label{e-prodchi}
   \chi(a_1,\ldots,a_t)=\prod_{i=1}^t\chi_i(a_i) \text{ for }(a_1,\ldots,a_t)\in R.
\end{equation}
Recall also that $\rad(R)=\rad(R_1)\times\ldots\times\rad(R_t)$ and
$\soc(R)=\soc(R_1)\times\ldots\times\soc(R_t)$.
Finally,
$R=\soc(R)\Longleftrightarrow\rad(R)=\{0\}\Longleftrightarrow \rad(R_i)=\{0\}\text{ for all }i
\Longleftrightarrow R$ is semisimple.
This last case will be of particular interest to us.
\end{rem}

\medskip
The above leads to the following identity for the values of the homogenous weight.
A similar formula, but in terms of the M\"obius function and for the case where all~$R_i$ are commutative principal ideal rings,
appeared in~\cite[Thm.~4.1]{FaLi10}, see~\cite[p.~4]{FLL13}.
\begin{prop}\label{P-homogProduct}
Let~$R=R_1\times\ldots\times R_t$ be a direct product of (not necessarily local) Frobenius rings.
For $i\in[t]$ let $\omega_i$ be the normalized homogeneous weight on~$R_i$.
Then the normalized homogeneous weight on~$R$ is given by
\[
   \omega(a_1,\ldots,a_t)=1-\prod_{i=1}^t\big(1-\omega_i(a_i)\big) \text{ for }(a_1,\ldots,a_t)\in R.
\]
\end{prop}
\begin{proof}
Let~$\chi_i$ be a generating character of~$R_i$.
Then Remark~\ref{R-ProdFrob} provides us with a generating
character~$\chi$ of~$R$.
Using~\eqref{e-KrawHom} and $|R^*|=\prod_{i=1}^t |R_i^*|$, we compute
\begin{align*}
 \sum_{(u_1,\ldots,u_t)\in R^*}\chi(a_1u_1,\ldots,a_tu_t)
       &=\sum_{u_1\in R_1^*,\ldots,u_t\in R_t^*}\prod_{i=1}^t\chi_i(a_iu_i)
        =\prod_{i=1}^t\sum_{u\in R_i^*}\chi_i(a_iu)=\prod_{i=1}^t|R_i^*|\big(1-\omega_i(a_i)\big)\\
     &=|R^*|\prod_{i=1}^t\big(1-\omega_i(a_i)\big).
\end{align*}
Now~\eqref{e-HomWt} leads to the desired identity.
\end{proof}

This result brings us immediately to an explicit formula for the homogeneous weight in the following case.

\begin{prop}\label{P-HomogProdq}
Let $R=R_1\times\ldots\times R_t$, where each~$R_i$ is a finite local Frobenius ring
with residue field $R_i/\rad(R_i)$ of order~$q$.
Then
\[
     \omega(a)=\left\{\begin{array}{cl}{\DS 1-\Big(\frac{-1}{q-1}\Big)^{\wt(a)}}
         &\text{if }a\in\soc(R),\\[1.7ex]
        1,&\text{otherwise,}\end{array}\right.
\]
where $\wt(a):=|\{i\mid a_i\neq 0\}|$ denotes the Hamming weight of $a=(a_1,\ldots,a_t)$.
\end{prop}
\begin{proof}
This is a consequence of Example~\ref{E-LocFrob} along with
Proposition~\ref{P-homogProduct} and the fact that $\soc(R)=\soc(R_1)\times\ldots\times\soc(R_t)$.
\end{proof}

In the same way we can compute the homogeneous weight on any finite Frobenius ring that is given as
as a direct product of local rings.
We will need to keep track of the orders of the residue fields of the component rings and thus fix the following notation.
For the rest of this paper, let~$R$ be a finite Frobenius ring of the form
\begin{equation}\label{e-Rform}
\left.
\begin{array}{l}
  R=R_1\times\ldots\times R_t, \text{ where }R_i=R_{i,1}\times\ldots\times R_{i,n_i}\text{ with $R_{i,j}$ local} ,\\[1ex]
  |R_{i,j}/\rad(R_{i,j})|=|\soc(R_{i,j})|=q_i\text{ for all }j\in[n_i]\text{ and }q_1,\ldots,q_t\text{ distinct.}
\end{array}
\quad\right\}
\end{equation}
Recall that $\soc(R_{i,j})\cong R_{i,j}/\rad(R_{i,j})$.

Propositions~\ref{P-homogProduct} and~\ref{P-HomogProdq} yield the following generalization of Example~\ref{E-LocFrob}.
\begin{theo}\label{T-HomWtGen}
Let~$R$ be as in~\eqref{e-Rform} and write its elements as $a=(a_1,\ldots,a_t)$, where $a_i\in R_i$.
Using the Hamming weight~$\wt$ on each~$R_i$, the homogeneous weight on~$R$ is given by
\[
   \omega(a_1,\ldots,a_t)=\left\{\begin{array}{cl}
      1-\prod_{i=1}^t\big(\frac{-1}{q_i-1}\big)^{\wt(a_i)},&\text{if }a\in\soc(R),\\[1ex] 1,&\text{otherwise}.\end{array}\right.
\]
\end{theo}

We close this section with the following immediate insight about the induced partition.

\begin{rem}\label{R-HomWtValues}\
\begin{alphalist}
\item $\omega(a)\neq1$ for all $a\in\soc(R)$. Thus $R\,\backslash\,\soc(R)$ is a block of $\cPhom$.
        This generalizes the situation for local Frobenius rings in Example~\ref{E-LocFrob}.
\item There exists a nonzero $a\in R$ such that $\omega(a)=0$ if and only if there exists an~$i\in[t]$ such that $q_i=2$
        and $n_i\geq2$.
        In this case,~$\cPhom$ is not reflexive due to Remark~\ref{R-Dual0}(a).
        Example~\ref{E-Phom}(b) is the smallest such ring.
\end{alphalist}
\end{rem}

A detailed study of the homogeneous partition and its dual is presented in the next section.

\section{The Partition Induced by the Homogeneous Weight}\label{SS-HomogPart}
We characterize the rings for which the partition~$\cPhom$ induced by the homogeneous weight
is reflexive or even self-dual.
Thanks to Proposition~\ref{P-leftrightbidual} we may and will restrict ourselves to the right dual of $\cPhom$ in order to study
self-duality (and reflexivity).
In the reflexive case we also determine the right dual partition and the Krawtchouk coefficients explicitly.

Recall from Remark~\ref{R-UnitOrbits} that the right dual partition of $\cPhom$ in the
sense of Definition~\ref{D-DualPart}(b) does not depend
on the choice of the generating character and is denoted by~$\wcPhomr$.

Some basic properties of~$\cPhom$ were presented already in Remark~\ref{R-HomWtValues}.
We now focus on a simple consequence of Theorem~\ref{T-HomWtGen} that will turn out to be crucial for characterizing reflexivity.
Recall that~$R$ is as in~\eqref{e-Rform}.
Theorem~\ref{T-HomWtGen} shows that the homogeneous weight of $a=(a_1,\ldots,a_t)\in R$ depends on the Hamming weights $\wt(a_i)$.
In particular, if $a,b\in\soc(R)$ are such that $\wt(a_i)=\wt(b_i)$ for all $i\in[t]$, then $\omega(a)=\omega(b)$.
This shows that the homogeneous partition is closely related to the product of the Hamming partitions
on~$\soc(R_i),\,i\in[t]$.

We will therefore study this product partition first and come back to the homogeneous partition thereafter.
As we will show later, the homogeneous partition is reflexive if and only if it coincides on the socle with the product of the
Hamming partitions.

Denote by~$\cH_i$ the Hamming partition on~$\soc(R_i)=\soc(R_{i,1})\times\ldots\times\soc(R_{i,n_i})$, thus
$\cH_i=P_{i,0}\mmid\ldots\mmid P_{i,n_i}$ with  blocks
$P_{i,j}=\{(a_{i,1},\ldots,a_{i,n_i})\in \soc(R_i)\mid \wt(a_{i,1},\ldots,a_{i,n_i})=j\}$.
The induced product partition on $\soc(R_1)\times\ldots\times\soc(R_t)=\soc(R)$ is given by
\begin{equation}\label{e-cH}
   \cH:=\cH_1\times\ldots\times\cH_t,
\end{equation}
see~\eqref{e-prodPart}, and consists of the blocks
\begin{equation}\label{e-Pm}
    P_m:=P_{1,m_1}\times\ldots\times P_{t,m_t}=\{(a_1,\ldots,a_t)\in\soc(R)\mid \wt(a_i)=m_i\text{ for all }i\in[t]\},
\end{equation}
where $m:=(m_1,\ldots,m_t)\in\cM:=[n_1]_0\times\ldots\times[n_t]_0$.

For the dual partition the following identifications are useful.
Note first that $R/\rad(R)\cong R_1/\rad(R_1)\times\ldots\times R_t/\rad(R_t)$
for~$R$ as in~\eqref{e-Rform}.
In the same way
\begin{equation}\label{e-Rirad}
  R_i/\rad(R_i)\cong R_{i,1}/\rad(R_{i,1})\times\ldots\times R_{i,n_i}/\rad(R_{i,n_i}).
\end{equation}
This allows us to consider the Hamming weight on $R_i/\rad(R_i)$.
For $a_i=(a_{i,1},\ldots,a_{i,n_i})\in R_i$ put
\begin{equation}\label{e-Hammrad}
    \wt\big(a_i+\rad(R_i)\big):=|\{j\mid a_{i,j}\not\in \rad(R_{i,j})\}|. 
\end{equation}

Now we can formulate the following  duality.

\begin{theo}\label{T-HammRi}
Consider $\cH=(P_m)_{m\in\cM}$ as above. Define the partition~$\cH'$ of~$R$ as
\[
       \cH'=\big(P_m)_{m\in\cM\cup\{\diamond\}}, \text{ where }P_\diamond:=R\,\backslash\,\soc(R).
\]
The right dual partition~$\widehat{\cH'}^{\scriptscriptstyle[\chi,r]}$ does not depend on~$\chi$ and coincides with the left
dual partition $\widehat{\cH'}^{\scriptscriptstyle[\chi,l]}$.
It is given by
\[
   \widehat{\cH'}^{\scriptscriptstyle[\chi,r]}=\big(Q_m)_{m\in\cM\cup\{\diamond\}},
\]
where $Q_\diamond=\rad(R)\backslash\{0\},\ Q_{0}=\{0\}$ and for $m=(m_1,\ldots,m_t)\in\cM\backslash\{0\}$
\[
   Q_m=\{(a_1,\ldots,a_t)\in R\,\backslash\,\rad(R)\mid \wt(a_i+\rad(R_i))=m_i\text{ for all }i\in[t]\}.
\]
In particular,~$\cH'$ is reflexive.
If $R=\soc(R)$ then $|\cH'|=s:=\prod_{i=1}^t(n_i+1)$, and~$\cH'$ is simply the
product of the Hamming partitions and thus self-dual.
If $R\neq\soc(R)$ then $|\cH'|=s+1$.
\end{theo}
\begin{proof}
We make use of Theorem~\ref{T-IndProdPart} and Proposition~\ref{P-PartSubgroup}.
For this we consider $\soc(R)$ as a subgroup of~$R$.
Then~$\cH$ is a partition of this subgroup, and it is given as the product of the Hamming partitions~$\cH_i$
of $\soc(R_i)$.
By Example~\ref{E-HammPart}, the dual partitions $\widehat{\cH_i}$ are the Hamming partitions on
the character groups $\widehat{\soc(R_i)}$, and Theorem~\ref{T-IndProdPart} implies that $\widehat{\cH}$ is the partition
$\widehat{\cH_1}\times\ldots\times\widehat{\cH_t}$ of $\widehat{\soc(R)}=\widehat{\soc(R_1)}\times\ldots\times\widehat{\soc(R_t)}$.
Proposition~\ref{P-PartSubgroup} yields that the partition~$\widehat{\cH'}$ of the group $\widehat{R}$
consists of the blocks
$\soc(R)^{\perp}\backslash\{\veps\},\ \{\veps\}$, and the blocks
\begin{equation}\label{e-Rhatright}
  \{(\chi_1\!\cdot\!a_1,\ldots,\chi_t\!\cdot\!a_t)\in\widehat{R}\,\backslash\,\soc(R)^{\perp}\mid
   \wt((\chi_i\!\cdot\!a_i)_{|\text{soc}(R_i)})=m_i\text{ for all }i\in[t]\}
\end{equation}
for all $m\in\cM\backslash\{0\}$, and where~$\chi_i$ is a fixed generating character of~$R_i$ (see also Remark~\ref{R-ProdFrob}).
Now the isomorphism $\alpha_r$ from~\eqref{e-RRnhat} turns the
partition~$\widehat{\cH'}$ of~$\widehat{R}$ into the partition $\widehat{\cH'}^{\scriptscriptstyle[\chi,r]}$ of~$R$.
Since $\alpha_r^{-1}(\soc(R)^{\perp})=\text{ann}_l(\soc(R))=\rad(R)$ due to~\eqref{e-Cdual}, it is clear that the partition
$\widehat{\cH'}^{\scriptscriptstyle[\chi,r]}$ consists of the blocks $\rad(R)\backslash\{0\},\;\{0\}$, and the sets~$Q_m$
given in the theorem.
All of this proves the desired duality.
Evidently, the left-sided version of~\eqref{e-Rhatright} is true as well, and the analogous proof establishes
$\widehat{\cH'}^{\scriptscriptstyle[\chi,r]}=\widehat{\cH'}^{\scriptscriptstyle[\chi,l]}$.
The cardinality of $\cH'$ is clear from $|\cM|=s$, and reflexivity follows from Theorem~\ref{T-ReflCrit}.
Finally, if $R=\soc(R)$, then $\rad(R)=\{0\}$, so that in this case the blocks~$P_\diamond$ and $Q_\diamond$ are missing,
and the sets~$Q_m$ are the blocks of the product partition~$\cH$ on $\soc(R)=R$.
Therefore~$\cH'$ is self-dual.
\end{proof}

From now on we will simply write $\widehat{\cH'}$ for the dual partition.

\begin{cor}\label{C-HH'Kraw}
The Krawtchouk coefficients of the pair $(\cH',\wcHr)$ fromTheorem~\ref{T-HammRi} are given by
\[
    K_{\ell,m}=\left\{\begin{array}{cl}
        -|\soc(R)|=-\prod_{i=1}^tq_i^{n_i},  &\text{if }m=\diamond=\ell\\
        |R\,\backslash\,\soc(R)|,&\text{if }m=\diamond,\,\ell=0,\\
        0,&\text{if } m=\diamond,\,\ell\notin\{0,\diamond\}\\
        |P_{1,m_1}\times\ldots\times P_{t,m_t}|=\prod_{i=1}^t\binom{n_i}{m_i}(q_i-1)^{m_i},&\text{if }m\neq\diamond=\ell\\
        \prod_{i=1} ^tK_{m_i}^{(n_i,q_i)}(\ell_i),&\text{if }m\neq\diamond\neq\ell,
      \end{array}\right.
\]
for all $\ell,m\in\cM\cup\{\diamond\}$, and
where $K_{m_i}^{(n_i,q_i)}(\ell_i)$ are the Krawtchouk coefficients of the Hamming partition of
$\soc(R_{i,1})\times\ldots\times\soc(R_{i,n_i})$.
In the special case where $R=\soc(R)$, and thus each $R_{i,j}$ is a field, only the last case occurs.
\end{cor}
\begin{proof}
This is a consequence of Proposition~\ref{P-PartSubgroup} along with Theorem~\ref{T-IndProdPart} and the classical
Krawt\-chouk
coefficients for the Hamming partition given in Example~\ref{E-HammPart}.
\end{proof}

Now we can return to the homogeneous weight.
Note first that the equivalence relation corresponding to the partition~$\cH'$ is given by
\begin{equation}\label{e-H'equ}
   a\widesim_{\cH'}b\Longleftrightarrow \left\{\begin{array}{l}
               a,b\in R\,\backslash\,\soc(R) \text{ or} \\
               a,b\in\soc(R)\text{ and }\wt(a_i)=\wt(b_i)\text{ for all }i\in[t].\end{array}\right.
\end{equation}
A comparison to the homogeneous weight in Theorem~\ref{T-HomWtGen} shows that
$\cH'\leq\cPhom$.
In other words, $a\widesim_{\cH'}b\Longrightarrow \omega(a)=\omega(b)$ for all $a,b\in R$.
The converse of this implication is not true in general.
In other words, $\cPhom$ may be strictly coarser than~$\cH'$.
Theorem~\ref{T-HomWtGen} indicates that the particular values of~$q_1,\ldots,q_t$ decide on the difference between
these two partitions.
We cast the following definition.
It is simply made to reflect the case where $\cH'=\cPhom$, as we will show in Theorems~\ref{T-PhomAdmiss} 
and~\ref{T-PhomNonAdmiss}.

\begin{defi}\label{D-AdmPrimes}
Given the list $\cL:=[(q_1,n_1),\ldots,(q_t,n_t)]$ of distinct prime powers~$q_i$ and multiplicities $n_i\in\N$.
Then~$\cL$ is called \emph{separating} if the following holds:
\\
whenever $m:=(m_1,\ldots,m_t),\,\ell:=(\ell_1,\ldots,\ell_t)\in\cM=[n_1]_0\times\ldots\times[n_t]_0$ and $m\neq\ell$, then
\[
  \prod_{i=1}^t(q_i-1)^{m_i}\neq  \prod_{i=1}^t(q_i-1)^{\ell_i}\   \text{ or }\
  \sum_{i=1}^t m_i\not\equiv\sum_{i=1}^t\ell_i~\mod 2.
\]
We call a list $[q_1,\ldots,q_t]$ of distinct prime powers \emph{separating}, if $[(q_1,1),\ldots,(q_t,1)]$ is separating.
An integer~$N$ is \emph{separating} if its list of distinct prime factors is separating.
\end{defi}

As we will see below, for a separating list~$\cL$ the partition~$\cH'$ separates the
elements of~$R$ according to their homogeneous weight.

The list $[2,3,7,13]$ is not separating, and the condition is violated in two ways:
$(2-1)(3-1)(7-1)=(13-1)$ and $(3-1)(7-1)=(2-1)(13-1)$.
The list $[3,7,13]$ is separating.
The list $[(2,1),(3,2),(5,1)]$ is not separating because $(3-1)^2=(2-1)(5-1)$.
The list $[(3,2),(5,1)]$ is separating, but $[(3,4),(5,2)]$ is not.
The case $m=0$ shows that if $[(q_1,n_1),\ldots,(q_t,n_t)]$ is separating and $q_i=2$ for some $i$, then $n_i=1$.
Finally, if $[q_1,\ldots,q_t]$ is not separating, then $t\geq4$.

Starting with three distinct primes and using sufficiently large primes one easily shows that there
exist infinitely many separating integers.

\begin{theo}\label{T-PhomAdmiss}
Let $R$ be as in~\eqref{e-Rform}.
Suppose the list $\cL=[(q_1,n_1),\ldots,(q_t,n_t)]$ is separating.
Then $\cPhom=\cH'$, where~$\cH'$ is as in Theorem~\ref{T-HammRi}.
As a consequence, $\cPhom$ is reflexive and $\wcPhoml=\wcPhomr=\wcHr$.
Moreover, if $R=\soc(R)$, i.e., $R$ is semisimple,
then the homogeneous partition coincides with the product of the Hamming partitions~$\cH_i$ on each~$R_i$,
and thus is self-dual.
\end{theo}

\begin{proof}
First, the separating property guarantees that if $q_i=2$ then $n_i=1$.
With Remark~\ref{R-HomWtValues}(b) we conclude that $\omega(a)\neq0$ for all $a\neq0$.
Thus $\{0\}$ is a block of~$\cPhom$.
Moreover, $R\,\backslash\,\soc(R)$ is a block of~$\cPhom$, see Remark~\ref{R-HomWtValues}(a).
Since both sets are also blocks of~$\cH'$ and $\cH'\leq\cPhom$,  it remains to show that for all $a,b\in\soc(R)$ such that
 $\omega(a)=\omega(b)$ we have $a\widesim_{\cH'}b$.
By Theorem~\ref{T-HomWtGen} $\omega(a)=\omega(b)$ yields
$(-1)^{\wt(b)}\prod_{i=1}^t(q_i-1)^{\wt(a_i)}=(-1)^{\wt(a)}\prod_{i=1}^t(q_i-1)^{\wt(b_i)}$, where
$\wt(a)=\sum_{i=1}^t\wt(a_i)$ and similarly for~$b$.
Since $q_i-1>0$, this leads to $\wt(a)\equiv\wt(b)~\mod 2$ and
$\prod_{i=1}^t(q_i-1)^{\wt(a_i)}=\prod_{i=1}^t(q_i-1)^{\wt(b_i)}$.
Since~$\cL$ is separating, this  yields $\wt(a_i)=\wt(b_i)$ for all $i\in[t]$, and therefore $a\widesim_{\cH'}b$.
This concludes the proof.
\end{proof}

Before turning to the non-separating case, let us present some examples of the homogeneous weight on integer residue rings.
We use the notation from Theorem~\ref{T-HammRi}.
The last result allows us to simply write $\wcPhom$ for the dual partitions.

\begin{exa}\label{E-homogWt2}
\begin{alphalist}
\item Let~$t=1,\,n_1=1$, and $q_1=q$ (so $[(q,1)]$ is separating), thus~$R$ is a local Frobenius ring with $|R/\rad(R)|=q$.
        Then $\cM=\{0,1\}$ and  the partitions $\cPhom=\cH'= P_{0}\mmidbig P_{1}\mmidbig P_\diamond$ and
        $\wcPhom=\widehat{\cH'}=Q_{0}\mmidbig Q_\diamond\mmidbig Q_{1}$ read as
        \[
          \cPhom= \{0\}\mmidbig \soc(R)\backslash\{0\}\mmidbig R\,\backslash\,\soc(R)
          \ \text{ and }
            \wcPhom=\{0\}\mmidbig \rad(R)\backslash\{0\} \mmidbig R^*.
        \]
        The Krawtchouk matrix, indexed row- and columnwise by the partition sets of~$\wcPhom$ and $\cPhom$ in the given
        order, is
      \[
           \begin{pmatrix} 1&q-1&|R|-q\\1&q-1&-q\\1&-1&0\end{pmatrix}.
      \]
      For $R=\Z_8$, this matrix also appears in~\cite[p.~1553]{Cam98} by Camion.
\item Let $R=\Z_{p^{r}}\times\Z_{q^{s}}$, where $p,\,q$ are distinct primes and
       $r,s\geq1$.
       The list $[p,q]$ is separating.
       Note that $R\cong\Z_N$, where $N=p^rq^s$, and hence~$N$ is separating.
       For the component rings, socle and radical satisfy $\soc(\Z_{p^{r}})\backslash\{0\}=\cO_{p^{r-1}}$ (the $(\Z_{p^{r}})^*$-orbit)
       and  $\rad(\Z_{p^{r}})=(p)$ and analogously for~$\Z_{q^{s}}$.
       The homogeneous partition $\cPhom$ and the values of the homogeneous weight are given by
       (in the order $P_{(0,0)}\mmid P_{(1,0)}\mmid P_{(0,1)}\mmid P_{(1,1)}\mmid P_\diamond$)
        \[
          \begin{array}{c||c|c|c|c|c}
             \cPhom&0&\cO_{p^{r-1}}\times\{0\}& \{0\}\times\cO_{q^{s-1}}
                 & \cO_{p^{r-1}}\times\cO_{q^{s-1}}
                 & R\,\backslash(p^{r-1})\times(q^{s-1})\phantom{\Big|} \\\hline
              \omega &0& \frac{p}{p-1}    &\frac{q}{q-1}       &\frac{pq-p-q}{(p-1)(q-1)} & 1\phantom{\Big|}
          \end{array}
      \]
      The dual partition is
      $\wcPhom=0\mmidbig (p)\times(q)\backslash\{0\}\mmidbig (p)\times \Z_{q^{s}}^*\mmidbig \Z_{p^{r}}^*\times(q)
         \mmidbig \Z_{p^{r}}^*\times\Z_{q^{s}}^*$, and the Krawtchouk matrix has the form
      \[
          K=\begin{pmatrix} 1&p-1&q-1&(p-1)(q-1)&p^rq^s-pq\\
                                1&p-1&q-1&(p-1)(q-1)&-pq\\
                                1&p-1&     -1 &1-p              & 0\\
                                1&  -1    &q-1&1-q              & 0\\
                                1&   -1   &    -1 &      1      & 0 \end{pmatrix}.
      \]
      If $N=pq$, the last block of~$\cPhom$ and the second block of~$\wcPhom$  are missing, and so are
      the last column and second row of~$K$.
      For this case, the values of the homogeneous weight also appears in~\cite[Ex.~3]{Byr11} by Byrne.
\end{alphalist}
\end{exa}

Here is the smallest non-separating integer~$N$.
As one may expect, the homogeneous partition is not reflexive.
\begin{exa}\label{E-homogWt3}
Consider~$\Z_N$, where~$N$ is the non-separating integer $N=2\cdot3\cdot7\cdot13=546$.
Note that $\Z_N\cong\Z_2\times\Z_3\times\Z_7\times\Z_{13}$ and since~$n_i=1$ for all~$i$, the partition $\cH'$
is simply the orbit partition under the multiplicative action of~$\Z_N^*$.
Hence~$|\cH'|=2^4=16$.
But the elements in the blocks $P_{(0,1,1,0)}$ and $P_{(1,0,0,1)}$ (for the notation see~\eqref{e-Pm}) have the same
homogeneous weight $11/12$, and similarly the elements in $P_{(0,0,0,1)}\cup P_{(1,1,1,0)}$ have weight $13/12$.
This leads to~$|\cPhom|=14$.
On the other hand, using a computer algebra program one computes that the dual partition~$\wcPhom$ is exactly the orbit partition~$\cH'$
and thus $|\wcPhom|=16$. Hence~$\cPhom$ is not reflexive.
\end{exa}

The situation of the last example is true for all non-separating cases.

\begin{theo}\label{T-PhomNonAdmiss}
Let~$R$ be as in~\eqref{e-Rform} and assume that $\cL=[(q_1,n_1),\ldots,(q_t,n_t)]$ is non-separating.
Then the partition $\cPhom$ is not reflexive.
More precisely, $\cPhom$ is strictly coarser than the partition~$\cH'$ from Theorem~\ref{T-HammRi}
whereas $\wcPhomr=\wcHr$.
Thus $|\wcPhomr|>|\cPhom|$.
\end{theo}
\begin{proof}
As in Theorem~\ref{T-HammRi} we denote the blocks of~$\cH'$ by $P_m,\,m\in\cM\cup\{\diamond\}$.
Let $K_{\ell,m}$ be the Krawtchouk coefficients of the pair $(\cH',\wcHr)$.
It is clear from~\eqref{e-H'equ} that $\cPhom\geq\cH'$ and thus $\wcPhomr\geq\wcHr$ due to Remark~\ref{R-Dual0}(b).
Furthermore, $P_{\diamond}=R\,\backslash\,\soc(R)$ is a block of~$\cPhom$.
Theorem~\ref{T-HomWtGen} shows that for $m\in\cM$
\begin{equation}\label{e-PmWeight}
   a\in P_m\Longrightarrow \omega(a)=1-\prod_{i=1}^t\big(\frac{-1}{q_i-1}\big)^{m_i}.
\end{equation}
From the fact that~$\cL$ is non-separating it follows that there exist two distinct partition sets $P_m,\,P_{m'}$ for some
$m,\,m'\in\cM$ such that all elements of $P_m\cup P_{m'}$ have the same homogeneous weight.
Thus $\cPhom>\cH'$.

We have to show that $\wcPhomr=\wcHr$.
In order to do so, we need some preparation.
Evidently, any partition set~$P$ of $\cPhom$, where $P\neq P_{\diamond}$, is of the form $P=\cup_{m\in\cL}P_m$ for a subset $\cL\subseteq\cM$.
Then
\begin{equation}\label{e-SumKlm}
       \sum_{b\in P}\chi(ab)= \sum_{m\in\cL}\sum_{b\in P_m}\chi(ab)=\sum_{m\in\cL}K_{\ell,m}\text{ for any }a\in Q_\ell,
\end{equation}
where $\wcHr=(Q_{\ell})_{\ell\in\cM\cup\{\diamond\}}$; see~\eqref{e-klm2r}.

Let $\ell,\,\ell'\in\cM\cup\{\diamond\}$ be such that $Q_\ell$ and $Q_{\ell'}$ are contained in the same partition set of
$\wcPhomr$.
We have to show that $\ell=\ell'$.
Let us first convince ourselves that we may assume that~$\ell,\ell'$ are in $\cM\backslash\{0\}$ for otherwise we are done.
To do so, note that  $Q_{0}=\{0\}$ is a partition set of~$\wcPhomr$ as this is a general property of dual partitions, see
Remark~\ref{R-Dual0}(a).
Thus $\ell\neq0\neq\ell'$.
Next, recall that the index~$\diamond$ occurs only if $R\neq\soc(R)$ and that in this case
$P_{\diamond}$ is a block of~$\cPhom$.
Corollary~\ref{C-HH'Kraw} shows that $K_{\ell,\diamond}=0$ for all $\ell\not\in\{0,\diamond\}$ whereas
$K_{\diamond,\diamond}<0$ and $K_{0,\diamond}>0$.
This implies that $Q_\diamond$ must also be a block of~$\wcPhomr$.
All of this shows that we may assume $\ell,\ell'\in\cM\backslash\{0\}$.

With the aid of~\eqref{e-SumKlm} our assumption on~$\ell,\,\ell'$ may be written as
\begin{equation}\label{e-KJIsum}
     \sum_{m\in\cL}K_{\ell,m} =\sum_{m\in\cL}K_{\ell',m} \text{ for all partition sets }\bigcup_{m\in\cL}P_m\text{ of }\cPhom.
\end{equation}

We will make use of the Krawtchouk coefficients for the case where $m=e_i=(0,\ldots,1,\ldots,0)$ (with~$1$ in the $i$th position).
From Corollary~\ref{C-HH'Kraw} we have
$K_{\ell,e_i}=K_1^{(n_i,q_i)}(\ell_i)=(n_i-\ell_i)q_i-n_i$, which along with~\eqref{e-KJIsum} results in the implication
\begin{equation}\label{e-Impl}
   \text{$P_{e_i}$ is a block of $\cPhom$}\Longrightarrow \ell_i=\ell'_i.
\end{equation}
Hence we aim at showing that the sets $P_{e_i}$ are blocks of~$\cPhom$.
In order to do so, we assume without loss of generality that
\begin{equation}\label{e-sortq}
      q_1<\ldots<q_t.
\end{equation}
\underline{Case~1:}
Suppose that for all $i\in\{1,\ldots,t\}$ we have $q_i-1\neq\prod_{j=1}^t(q_j-1)^{m_j}$ whenever $\sum_{j=1}^t m_j$ is odd.
Then~\eqref{e-PmWeight} shows that~$P_{e_i}$ is a block of~$\cPhom$, and hence~\eqref{e-Impl} implies $\ell=\ell'$.
\\[.6ex]
\underline{Case 2:}
Let $i\in\{1,\ldots,t\}$ be minimal such that $q_i-1=\prod_{j=1}^t(q_j-1)^{m_j}$ for some $m\neq e_i$ and $\sum_{j=1}^t m_j$ odd.
Then $m=(m_1,\ldots,m_{i-1},0,\ldots,0)$.
Let $m^{(1)},\ldots,m^{(s)}$ be all such indices satisfying
\[
     q_i-1=\prod_{j=1}^t(q_j-1)^{m^{(r)}_j}\text{ for }r=1,\ldots,s.
\]
Then the same argument as in Case~1 along with~\eqref{e-sortq} shows that the sets $P_{e_a}$, where $a<i$, are
blocks of $\cPhom$.
Hence by~\eqref{e-Impl}
\begin{equation}\label{e-ellell'}
   \ell_a=\ell'_a\text{ for }a=1,\ldots,i-1.
\end{equation}
Next, the set $P':=\bigcup_{m\in\cL}P_m$, where $\cL=\{e_i,m^{(1)},\ldots,m^{(s)}\}$, is a block of $\cPhom$.
Since $m^{(r)}_j=0$ for $j\geq i$, the case $m\neq\diamond\neq\ell$ in Corollary~\ref{C-HH'Kraw} along with~\eqref{e-Krawtclassic}
shows that
\[
  K_{\ell,m^{(r)}}=K_{(\ell_1,\ldots,\ell_{i-1},0,\ldots,0),m^{(r)}}\text{ for all }r=1,\ldots,s
\]
and analogously for~$\ell'$.
Hence $K_{\ell,m^{(r)}}=K_{\ell',m^{(r)}}$ by~\eqref{e-ellell'} and thus
\[
  \sum_{m\in\cL}K_{\ell,m}=\sum_{r=1}^s K_{\ell,m^{(r)}}+K_{\ell,e_i}=\sum_{r=1}^s K_{\ell',m^{(r)}}+K_{\ell,e_i}.
\]
Thus~\eqref{e-KJIsum} yields $K_{\ell,e_i}=K_{\ell',e_i}$ and hence $\ell_i=\ell'_i$.

Now we may continue in the same fashion for the index set $\{i+1,\ldots,t\}$.
If there is no index $i'>i$ such that $q_{i'}-1=\prod_{j=1}^t(q_j-1)^{m_j}$ and $\sum_{j=1}^t m_j$ is odd, then we
may argue as in Case~1.
Otherwise, we choose the smallest $i'>i$ and argue as in Case~2.
Proceeding in this manner we finally arrive at $\ell=\ell'$, as desired.
\end{proof}

We close the paper with the following summary.

\begin{cor}\label{C-ZNPhom}
Let~$R$ be as in~\eqref{e-Rform}.
\begin{alphalist}
\item $\cPhom$ is reflexive if and only if it coincides on the socle with the product of the Hamming partitions on~$R_i,\,i=1,\ldots,t$, and this
      is the case if and only if $[(q_1,n_1),\ldots,(q_t,n_t)]$ is separating.
      In this case $\wcPhomr=\wcPhoml$.
      Moreover, $\cPhom$ is self-dual if and only if~$R$ is semisimple and $[(q_1,n_1),\ldots,(q_t,n_t)]$ is separating.
\item The homogeneous partition~$\cPhom$ on~$\Z_N$ is reflexive if and only if~$N$ is separating.
      The partition is self-dual if and only if $N$ is square-free and separating.
\end{alphalist}
\end{cor}

We leave it to future research to determine the values and the corresponding partition of the homogeneous weight for general Frobenius rings.

\bibliographystyle{abbrv}
\bibliography{literatureAK,literatureLZ}

\end{document}